\documentclass[letterpaper]{article}
\usepackage{times}
\usepackage{helvet}
\usepackage{courier}
\usepackage{url}
\usepackage{graphicx}
\title{Random Dictators with a Random Referee: Constant Sample Complexity Mechanisms for Social Choice}
\author{Brandon Fain\\
Duke University
\and
Ashish Goel\\
Stanford University
\and 
Kamesh Munagala\\
Duke University
\and
Nina Prabhu\\
Stanford University
}

\usepackage[utf8]{inputenc}
\usepackage{amsfonts}
\usepackage{indentfirst}
\usepackage{amsmath, amssymb, amsthm, amsfonts}
\usepackage{tikz, comment}
\usepackage{algpseudocode, algorithm}
\usepackage{tcolorbox, qtree}
\usepackage{enumerate}
\usepackage{physics}
\usepackage{fancyvrb}
\usepackage[margin=1in]{geometry}

\newcommand {\Se}{\mathcal S}

\newtheorem{theorem}{Theorem}
\newtheorem{corollary}{Corollary}

\newtheorem{lemma}{Lemma}
\newtheorem{definition}{Definition}

\begin{document}
\maketitle

\begin{abstract}
We study social choice mechanisms in an implicit utilitarian framework with a metric constraint, where the goal is to minimize \textit{Distortion}, the worst case social cost of an ordinal mechanism relative to underlying cardinal utilities. We consider two additional desiderata: \textit{Constant sample complexity} and \textit{Squared Distortion}. Constant sample complexity means that the mechanism (potentially randomized) only uses a constant number of ordinal queries regardless of the number of voters and alternatives. Squared Distortion is a measure of variance of the Distortion of a randomized mechanism.  

Our primary contribution is the first social choice mechanism with constant sample complexity \textit{and} constant Squared Distortion (which also implies constant Distortion). We call the mechanism \textit{Random Referee}, because it uses a random agent to compare two alternatives that are the favorites of two other random agents. We prove that the use of a comparison query is necessary: no mechanism that only elicits the top-k preferred alternatives of voters (for constant k) can have Squared Distortion that is sublinear in the number of alternatives. We also prove that unlike any top-k only mechanism, the Distortion of Random Referee meaningfully improves on benign metric spaces, using the Euclidean plane as a canonical example. Finally, among top-1 only mechanisms, we introduce \textit{Random Oligarchy}. The mechanism asks just 3 queries and is essentially optimal among the class of such mechanisms with respect to Distortion. 

In summary, we demonstrate the surprising power of constant sample complexity mechanisms generally, and just three random voters in particular, to provide some of the best known results in the implicit utilitarian framework.

\end{abstract}

\section{Introduction}
\label{sec:intro}
Consider the social choice problem of deciding on an allocation of public tax dollars to public projects. This is a voting problem over budgets. Clearly, the number of voters in such situations can be large. More interestingly, unlike in traditional social choice theory, there is no reason to believe that the number of alternatives (budgets) is small. It is therefore unreasonable to assume that we can elicit full ordinal preferences over alternatives from every agent. For a voting mechanism to be practical in such a setting, one would ideally like it to require only an absolute constant number of simple queries, regardless of the number of voters and alternatives. We call this property \textit{constant sample complexity}, and we explore mechanisms of this sort in this paper.
%\footnote{One might protest that such examples only highlight the need for a good candidate selection algorithm. In general, the problem of candidate selection is itself a social choice problem of similar scale and difficulty as the election problem. We see no strong reason to assume that such a pre-processing mechanism always exists. Alternatively, one can see the mechanisms we provide as just that: candidate selection algorithms.}  

We define our model more formally in Section~\ref{sec:formal}, but at a high level, we have a set $N$ of agents (or voters) and a set of alternatives $\Se$, from which we must choose a single outcome. We assume that $N$ and $\Se$ are both large, and that eliciting the full ordinal rankings may be prohibitively difficult.  Instead, we work with an ordinal query model, and a constant sample complexity mechanism uses only a constant number of these queries.

\begin{description}
\item[Top-$k$ Query.] ``What are your $k$ favorite alternatives, in order?'' (We call a top-$1$ query a \textbf{favorite} query); and 
\item[Comparison Query.]  ``Which of two given alternatives do you prefer?''  
\end{description} 

Query models are not just of theoretical interest. They can be used to reduce cognitive overload in voting. For example, in the context of Participatory Budgeting~\cite{PBP}, the space of possible budget allocations is large, and one mechanism is to ask voters to compare two proposed budgets. Similarly, in a context like transportation policy for a city, a single alternative can be an entire transportation plan. In such examples, not only are there many alternatives, but it may be infeasible to expect voters to compare more than two alternatives at the same time. We stress that constant sample complexity is particularly important in settings where there may be a large number of possibly complex alternatives.

To evaluate the quality of our mechanisms, we adopt the implicit utilitarian perspective with metric constraints~\cite{implicitUtilitarian,YuCheng,metricSocialChoiceRandomDictator,metricSocialChoiceLowerBounds,deterministic,Feldman}. That is, we assume that agents have cardinal costs over alternatives, and these costs are constrained to be metric, but asking agents to work with or report cardinal costs is impractical or impossible. We want to design social choice mechanisms to minimize the total social cost, but our mechanism is constrained only to use ordinal queries, that is, those that can be answered given a total order over alternatives. We therefore measure the efficiency of a mechanism as its \textit{Distortion} (see Section~\ref{sec:formal}), the worst case approximation to the total social cost.

%One might expect that constant sample complexity mechanisms that achieve low Distortion do not exist. However, prior to our work, the best known randomized mechanism for minimizing Distortion, even among mechanisms that take in full ordinal information, was the constant sample complexity \textit{Random Dictatorship} mechanism. The algorithm asks a single favorite query from an agent chosen uniformly at random, and has a tight Distortion bound of 3~\cite{metricSocialChoiceRandomDictator}. 

\section{Results} 

The starting point for our inquiry is the constant sample complexity \textit{Random Dictatorship} mechanism. The algorithm asks a single favorite query from an agent chosen uniformly at random, and has a tight Distortion bound of 3~\cite{metricSocialChoiceRandomDictator}. In this paper, we provide two new mechanisms (Random Referee and Random Oligarchy) that improve on this simple baseline in three different ways, outlined in each of our three technical sections. We hope that our work inspires future research on similarly lightweight mechanisms for social choice in large decision spaces.

\medskip
\noindent {\bf Random Referee: Comparison Queries and Squared Distortion.}
In Section~\ref{sec:bargain}, we show that one disadvantage of Random Dictatorship lies not in its Distortion, but in its variance. For randomized mechanisms, Distortion is measured as the expected approximation to the first moment of social cost. However, in many social choice problems, we might want a bound on the risk associated with a given mechanism. We capture this via \textit{Squared Distortion}, as suggested in~\cite{sequentialDeliberation}.  The Squared Distortion (see Definition~\ref{def:squaredDistortion} in Section~\ref{sec:formal}) is the expected approximation to the second moment of social cost. A mechanism with constant Squared Distortion has both constant Distortion and constant coefficient of variation of the Distortion.%\footnote{The Squared Distortion of a deterministic mechanism is simply the square of its Distortion.}  

We show that mechanisms using only top-$k$ queries (including Random Dictatorship) have Squared Distortion $\Omega(|\Se|)$.  This motivates us to expand our query model to incorporate information about the relative preferences of agents between alternatives, {\em i.e.} use comparison queries. We  define  a novel mechanism called \textit{Random Referee} (RR) that uses a random voter as a referee to compare the favorite alternatives of two other random voters (see Definition~\ref{def:RB} in Section~\ref{sec:bargain}). Our main result in Section~\ref{sec:bargain} is Theorem~\ref{theorem:squaredDistortionRandomBargain}: The Squared Distortion of RR is at most 21. This also immediately implies that the Distortion of RR is at most 4.583.%RR is the first constant sample complexity social choice mechanism shown to have constant Squared Distortion. Moreover, since the lowest known Squared Distortion of any deterministic social choice mechanism (Copeland) is 25~\cite{deterministic}, RR has the best bound on Squared Distortion of \textit{any} social choice mechanism, as far as we are aware. %We prove in Section~\ref{sec:bargain} that Random Referee has constant Squared Distortion. Random Referee is the first such randomized constant-complexity mechanism, as far as we know. 

\medskip
\noindent {\bf Random Referee: Euclidean Plane.} In Section~\ref{sec:euclidean}, we show that top-$k$ only mechanisms (again, including Random Dictatorship) achieve their worst case Distortion even on benign metrics such as low dimensional Euclidean spaces.  We analyze a special case on the Euclidean plane and prove that the Distortion of Random Referee beats that of any top-$k$ only mechanism. While the improvement we prove in Section~\ref{sec:bargain} is quantitatively small, it is qualitatively interesting: we demonstrate that by using a {\em single} comparison query, Random Referee can exploit the structure of the metric space to improve Distortion, whereas Random Dictatorship or any other top-$k$ only mechanism cannot.  We conjecture this result extends to Euclidean spaces in any dimension, and present some evidence to support this conjecture in Section~\ref{sec:open}.

\medskip
\noindent {\bf Random Oligarchy: Favorite Only Mechanisms.} In Section~\ref{sec:oligarchy}, we consider mechanisms that are restricted to favorite queries and show that constant complexity mechanisms are nearly optimal.  We present a mechanism that uses only three favorite queries  that  has Distortion at most 3 for arbitrary $|\Se|$; however, it also has Distortion for small $|\Se|$ that improves upon the best known favorite only mechanism from~\cite{2-agree}  that uses at most $\min(|N|+1, |\Se|+1)$ favorite queries. Comparing with a lower bound for favorite only mechanisms, Random Oligarchy has nearly optimal distortion and constant sample complexity. Though this mechanism does not have constant Squared Distortion like Random Referee, we present it to demonstrate again the surprising power of constant sample complexity randomized social choice mechanisms in general, and of queries to just three voters in particular.

\medskip
\noindent {\bf Techniques.}
We use different techniques to prove our different positive results. The proof of Squared Distortion (Theorem~\ref{theorem:squaredDistortionRandomBargain} in Section~\ref{sec:bargain}) relies heavily on Lemma~\ref{lemma:bargainTechnical}, in which we prove (essentially) that Random Referee chooses a low social cost alternative as long as  at least two of the three agents chosen at random are near the social optimum. 

The proof of Distortion for Euclidean spaces (Theorem~\ref{theorem:ordinalrandomBargain} in Section~\ref{sec:euclidean}) is the most technical result. We show that we can upper bound the Distortion of a mechanism by the worst case ``pessimistic distortion,'' of just a constant size tuple of points, where the ``pessimistic distortion'' considers all permutations of the points as participating in Random Referee and allows OPT to choose the optimal point on just this tuple. This allows us to employ a computer assisted analysis by arguing that if a high Distortion instance exists, we can detect it as an instance with a small constant number of points on a sufficiently fine (but finite) grid in the Euclidean plane. This approach may be of independent interest for providing tighter Distortion bounds for mechanisms in specific structured metric spaces.

%Finally, Theorem~\ref{theorem:oligarchy} in Section~\ref{sec:oligarchy} bounds the Distortion of Random Oligarchy by bounding the probability with which the mechanism outputs a given favorite alternative as a function of the number of agents who have that favorite alternative. The approach is similar to that taken in~\cite{2-agree}, but the results are stronger as we have already discussed.

%---------------------------------------%

\section{Related Work}

\noindent {\bf Distortion of Randomized Social Choice Mechanisms in Metrics.} The Distortion of randomized social choice mechanisms in metrics has been studied in~\cite{implicitUtilitarian,metricSocialChoiceRandomDictator,metricSocialChoiceLowerBounds,2-agree}. Of particular interest to us are the Random Dictatorship mechanism that uses a single favorite query and the 2-Agree mechanism~\cite{2-agree} that uses at most most $\min(|N|+1, |\Se|+1)$ favorite queries. Random Dictatorship has an upper bound on Distortion of 3~\cite{metricSocialChoiceRandomDictator}, and 2-Agree provides a strong guarantee on Distortion when $|\Se|$ is small (better than Random Dictatorship for $|\Se| \leq 6$). There is ongoing work on analyzing the Distortion of randomized ordinal mechanisms for other classic optimization problems like graph optimization~\cite{graphDistortion} and facility location~\cite{facilityLocation}.

\medskip
\noindent {\bf Squared Distortion and Variance.} We are aware of two papers in mechanism design that consider the variance of mechanisms for facility location on the real line~\cite{variance} and kidney exchange~\cite{varianceKidney}. Our work is more related to the former, but is not restricted to the real line, and does not focus on characterizing the tradeoff between welfare and variance. Using Squared Distortion as a proxy for risk was introduced in~\cite{sequentialDeliberation} along with the sequential deliberation protocol. Unlike sequential deliberation, Random Referee makes a constant number of ordinal queries. The most important baseline for Squared Distortion is the deterministic Copeland rule, which has Distortion $5$~\cite{deterministic} and therefore Squared Distortion 25. However, Copeland requires the communication of $\Omega(|N||\Se|\log(|\Se|))$ bits~\cite{communicationComplexity}, essentially the entire preference profile. Our Random Referee mechanism has constant sample complexity, and has better bounds on Squared Distortion (21) and Distortion (4.583) than the Copeland mechanism.

\medskip
\noindent {\bf Communication Complexity.} For a survey on the complexity of eliciting ordinal preferences to implement social choice rules, we refer the interested reader to~\cite{computationalSocialChoice}. Of particular interest to us is~\cite{communicationComplexity}, in which the authors comprehensively characterize the \textit{communication complexity} (in terms of the number of bits communicated) of common deterministic voting rules. A favorite query requires $O(\log(|\Se|))$ bits of communication, so our mechanisms have constant sample complexity, but logarithmic communication complexity.~\cite{sequentialElimination} and~\cite{votingCommunicationComplexity} design social choice mechanisms with low communication complexity when there are a small number of voters, but potentially a large number of alternatives. All of our mechanisms have guarantees that are independent of the number of voters.

\medskip
\noindent {\bf Strategic Incentives.} We do not consider truthfulness in this paper, and we do not use the term mechanism to imply any such property. While strategic incentives are not the focus of this work, we note that any truthful mechanism must have Distortion at least 3~\cite{Feldman}. Random Dictatorship has a Distortion of 3, and is therefore in some sense optimal among exactly truthful mechanisms. Other works suggest that truthfulness is also incompatible with the weaker notion of Pareto efficiency in randomized social choice~\cite{PESP,azizPESP}. Still other authors have considered the problem of truthful welfare maximization under range voting~\cite{rangeVoting} and threshold voting~\cite{thresholdVoting}.

\section{Preliminaries} 
\label{sec:formal}
We have a set $N$ of agents (or voters) and a set $\Se$ of alternatives, from which we must choose a single outcome. For each agent $u \in N$ and alternative $a \in \Se$, there is some underlying dis-utility $d_u(a)$. Let $p_u = \mbox{argmin}_{a \in \Se} d_u(a)$, that is, $p_u$ is the most preferred alternative for agent $u$. Ordinal preferences are specified by a total order $\sigma_u$ consistent with these dis-utilities (i.e., an alternative is ranked above another only if it has lower dis-utility). A preference profile $\sigma^{(N)}$ specifies the ordinal preferences of all agents. A deterministic social choice rule is a function $f$ that maps a preference profile $\sigma^{(N)}$ to an alternative $a \in \Se$. A randomized social choice rule maps a preference profile $\sigma^{(N)}$ to a distribution over $\Se$. 

We consider mechanisms that implement a randomized social choice rule using a constant number of queries of two types. A \textit{top}-$k$ \textit{query} asks an agent $u \in N$ for the first $k$ preferred alternatives according to the order $\sigma_u$ (ties can be broken arbitrarily). We refer to a top-1 query as a \textit{favorite query}, that asks an agent $u \in N$ for her most preferred alternative $p_u$. A \textit{top-$k$ only mechanism} $f$ uses only top-$k$ queries, for some constant $k$ (constant with respect to $|N|$ and $|\Se|$). Most of our lower bounds or impossibilities will be for any top-$k$ only mechanism (for constant $k$), whereas our positive results will only need favorite and comparison queries. A \textit{comparison query} with alternatives $a \in \Se$ and $b \in \Se$ asks an agent $u \in N$ for $\mbox{argmin}_{x \in \{a,b\}} d_u(x)$.

We use the term mechanism to clarify that our algorithms are in a query model. However, it is important to note that mechanisms so defined are still randomized social choice rules in the formal sense as long as they do not make queries based on exogenous information (e.g., names of participants). Our mechanisms will in fact be randomized social choice rules, and thus can be appropriately compared to other such rules in the literature that do not explicitly use a query model. By using the term mechanism, we do not mean to imply any strategic properties.

%\footnote{This special case where all bliss points are feasible alternatives corresponds to the 0-decisive case in~\cite{metricSocialChoiceRandomDictator, 2-agree}. The authors of~\cite{metricSocialChoiceRandomDictator, 2-agree} state in their work that it is a significant special case. Furthermore, our impossibility result in this model are only stronger, and our results for Random Oligarchy in Section~\ref{sec:oligarchy} carry through in the general setting with trivial adaptation. Our only results that need this assumption are those on Random Bargain in Section~\ref{sec:bargain}. The primary motivation for Random Bargain is to provide a low risk mechanism when the space of alternatives is very large, making this assumption relatively more innocuous. For example, in the context of electing a representative from among an entire population (something a simple mechanism can do without hitting communication complexity barriers), we would take the bliss point to be the voter herself.} 

\medskip
\noindent {\bf Distortion and Sample Complexity.} We measure the quality of an alternative $a \in \Se$ by its {\em social cost}, given by $ SC(a) = \sum_{u \in N} d_u(a)$. Let $a^* \in \Se$ be the minimizer of social cost. We define the commonly studied approximation factor called Distortion~\cite{distortion}, which measures the worst case approximation to the optimal social cost of a given mechanism. We use the expected social cost if $a$ is the outcome of a randomized mechanism, and we seek to minimize Distortion.

\begin{definition}
The \textbf{Distortion} of an alternative $a$ is $\mbox{Distortion}(a) = \frac{SC(a)}{SC(a^*)}.$ The Distortion of a social choice mechanism $f$ is $$ \mbox{Distortion}(f) = \sup_{\{d_u(a)\}} \mathbb{E}_{f \left(\sigma^{(N)} \right)} Distortion(a)$$ where $\sigma^{(N)}$ is a preference profile consistent with $\{d_u(a)\}$.
\end{definition}

We assume that $\Se$ is a set of points in a metric space such that dis-utility can be measured by the distance from an agent. Specifically, we assume there is a distance function $d: (N \cup \Se) \times ( N \cup \Se) \rightarrow \mathbb{R}_{\geq 0}$ satisfying the triangle inequality such that $d_u(a) = d(u, a)$. The metric assumption is common in the implicit utilitarian literature~\cite{metricSocialChoiceRandomDictator,metricSocialChoiceLowerBounds,sequentialDeliberation,2-agree,YuCheng,deterministic,Feldman}. It is also a natural assumption for capturing social choice problems for which there is a natural notion of distance between alternatives. For example, in our original motivating example of public budgets, there is are natural notions of distance between alternatives in terms of dollars.  

We do not assume access to $\sigma^{(N)}$ directly, which may be prohibitively difficult to elicit when there are many alternatives. Instead, we work with a query model. The queries are ordinal in the sense that they can be answered given only the information in $\sigma^{(N)}$. A mechanism $f$ has \textit{constant sample complexity} if there is an absolute constant $c$ such that for all $\Se$, $N$, and $\sigma^{(N)}$, $f$ can be implemented using at most $c$ queries. In this paper, we consider top-$k$ (and the special case of favorite) and comparison queries, and explore mechanisms with constant sample complexity.

\medskip
\noindent {\bf Squared Distortion.}  It is easy to see that randomization is necessary for constant sample complexity mechanisms to achieve constant Distortion. As $N$ grows large, any deterministic mechanism with constant sample complexity deterministically ignores (asks no queries of and receives no information from) an arbitrarily large fraction of $N$. An adversary can therefore place an alternative with 0 dis-utility for arbitrarily many agents; this gives a lower bound for Distortion approaching $N$ as $N$ becomes large.

This naturally leads us to ask: If we look at the distribution of outcomes produced by the mechanism, is this distribution well behaved?  Following~\cite{sequentialDeliberation}, we capture this notion via \textit{Squared Distortion}: essentially the approximation to the optimal second moment of social cost.
\begin{definition}
\label{def:squaredDistortion}
	The \textbf{Squared Distortion} of an alternative $a \in \Se$ is $Distortion^2(a) = \left(\frac{SC(a)}{SC(a^*)}\right)^2.$ The Distortion of a social choice mechanism $f$ is $$ Distortion^2(f) = \sup_{\{d_u(a)\}} \mathbb{E}_{f \left(\sigma^{(N)} \right)} Distortion^2(a)$$ where $\sigma^{(N)}$ is a preference profile consistent with $\{d_u(a)\}$.  
\end{definition} 

Note that a mechanism with constant Squared Distortion has both constant Distortion (by Jensen's inequality), and constant coefficient of variation.  One way to interpret having constant Squared Distortion is that the deviation of the social cost around its mean falls off quadratically instead of linearly, which means that the social cost of such a mechanism is well concentrated around its mean value. We note that this interpretation gives randomized social choice mechanisms with constant Squared Distortion an interesting application in candidate selection. In particular, one can imagine running such a mechanism (like our Random Referee) to generate a candidate list on which one can use a deterministic (but potentially complex) voting mechanism like Copeland. 

A related approach to understanding the distributional properties of a randomized mechanism is to characterize the tradeoff between Distortion (approximation to the first moment) and the variance of randomized mechanisms~\cite{variance}. Our specific goal in this paper is to develop a mechanism that achieves constant Distortion \textit{and} constant variance, and this combination is captured by having constant Squared Distortion. We leave characterizing the exact tradeoff between the quantities as an interesting open direction.

\section{Random Referee and Squared Distortion}
\label{sec:bargain}
%Constant sample complexity and Distortion are not the only desiderata we have for a randomized social choice mechanism. When considering mechanisms with constant query complexity, we also wish to bound the \textit{variance} involved in the randomization, which we analyze via Squared Distortion (Definition~\ref{def:squaredDistortion}). Unfortunately, Random Oligarchy and Random Dictatorship both have large Squared Distortion. 

%In this section, we first prove Theorem~\ref{theorem:squaredDIstortionLowerBound}: Though top-$k$ only mechanisms can have low Distortion, they all have Squared Distortion $\Omega(|\S|)$. This serves as the motivation to define a novel mechanism called Random Referee that uses  one random voter (referee) to perform a single comparison between the favorite alternatives of two random voters.

Our first result is Theorem~\ref{theorem:squaredDIstortionLowerBound}: Mechanisms that only elicit top-$k$ preferences, for constant $k$, must necessarily have Squared Distortion that grows linearly in the size of the instance. This holds even for mechanisms that elicit the top-$k$ preferences of all of the voters, mechanisms which would not have constant sample complexity. The proof is in the appendix.

\begin{theorem}
\label{theorem:squaredDIstortionLowerBound}
	Any top-$k$ only social choice mechanism has Squared Distortion $\Omega(|\Se|)$.
\end{theorem}

The problem with top-$k$ only mechanisms, exploited in the proof of Theorem~\ref{theorem:squaredDIstortionLowerBound}, is that they treat agents as indifferent between their $(k+1)$st favorite alternative and their least favorite alternative. This motivates the expansion of our query model to include \textit{comparison queries}. Recall that a comparison query with alternatives $a \in \Se$ and $b \in \Se$ asks an agent $u \in N$ for $\mbox{argmin}_{x \in \{a,b\}} d_u(x)$. We use a single comparison query in our Random Referee mechanism. 

\begin{definition}
\label{def:RB}
	The \textbf{Random Referee (RR)} mechanism samples three agents $u,v,w \in N$ independently and uniformly at random with replacement. $u$ and $v$ are asked for their favorite feasible alternatives $p_u$ and $p_v$ in $\Se$, and then $w$ is asked to compare $p_u$ and $p_v$. Output whichever of the two alternatives $w$ prefers.
\end{definition}

\medskip
\noindent \textbf{Upper Bound for Random Referee.} Our main result in this section is Theorem~\ref{theorem:squaredDistortionRandomBargain}: The Squared Distortion of RR is at most 21. This also implies that the Distortion of RR is at most 4.583. As far as we are aware, no randomized mechanism has lower Squared Distortion in general.

We will need the following technical lemma. Let $a^* \in \Se$ be the social cost minimizer. The basic intuition behind Random Referee is that it will choose a low social cost alternative as long as any two out of the three agents selected are near the optimal alternative $a^*$. Lemma~\ref{lemma:bargainTechnical} makes this intuition formal. For convenience, let $Z_u = d(u, a^*)$ for $u \in N$, i.e., the dis-utility of $a^*$ for agent $u$. Let $C(\{u, v\}, w) = \mbox{argmin}_{y \in \{p_u, p_v\}} d(w, y)$, that is, the alternative Random Referee outputs when $u$ and $v$ are selected to propose alternatives and $w$ chooses between them. 

\begin{lemma} 
\label{lemma:bargainTechnical}
For all $u, v, w, x \in N$,
	$$d(C(\{u,v\}, w), x) \leq Z_x + 2 \min \begin{cases} \max(Z_u, Z_v), \\ Z_w + \min(Z_u, Z_v) \end{cases}.$$
\end{lemma}
\begin{proof}
One should think of $u, v \in N$ as the agents drawn to present their favorite alternative, $w$ as the referee, and $x$ as the agent from whom we measure the dis-utility of the resulting outcome. The triangle inequality implies that $$d(C(\{u,v\}, w), x) \leq Z_x + d(a^*, C(\{u,v\}, w)).$$ We give two separate upper bounds on $d(a^*, C(\{u,v\}, w))$, thus yielding the min. We will frequently use the fact that for all $u \in N$, $d(u, a^*) \leq d(p_u, u) + d(u, a^*) \leq 2 Z_u$. This follows from the the definition of $p_u$: the favorite alternative of $u$, and thus no greater in distance from $u$ than $a^*$. The first bound is straightforward: $C(\{u,v\}, w) \in \{p_u, p_v\}$ by definition of Random Referee, so 
\begin{equation*}
\begin{split}
	d(a^*, C(\{u,v\}, w)) & \leq \max \left(d(a^*, p_u), d(a^*,p_v) \right) \\
	%& \leq \max \left(d(a^*, p_u) + d(p_u, p_u), d(a^*,p_v) + d(p_v, p_v) \right) \\
	& \leq 2 \max \left( Z_u, Z_v \right)
\end{split}
\end{equation*}	
In a sense, this bound concerns the situation where both $u$ and $v$ are near $a^*$, but $w$ is far away from $a^*$. Now we argue for the second bound. Suppose without loss of generality that $Z_u \leq Z_v$. $w$ chooses either $p_u$ or $p_v$. If $w$ chooses $p_u$, then $d(a^*, C(\{u,v\}, w)) = d(a^*, p_u) \leq 2 d(a^*, u) = 2 \min(Z_u, Z_v)$, and so the bound holds. If $w$ chooses $p_v$, then $d(w, p_v) \leq d(w, p_u)$, which implies 
\begin{equation*} 
\begin{split}
d(a^*, C(\{u,v\}, w))  & = d(a^*, p_v) \\
& \leq d(a^*, w) + d(w, p_u) \\
& \leq Z_w + d(w, a^*) + d(a^*, p_u) \\
& \leq 2 Z_w + 2Z_u \\
& = 2 \left( Z_w + \min(Z_u, Z_v)\right)
\end{split}
\end{equation*}
and again the bound holds. In either case, $d(a^*, C(\{u,v\}, w))$ is at most $2 \left(Z_w + \min(Z_u, Z_v) \right)$. Intuitively, this bound concerns the situation where $w$ and $u$ are close to $a^*$, but $v$ is far away from $a^*$. Taking the better of this bound with the $d(a^*, C(\{u,v\}, w)) \leq 2 \max(Z_u, Z_v)$ bound through the min and factoring out the 2 yields the lemma. 
\end{proof}

Using Lemma~\ref{lemma:bargainTechnical}, we can upper bound the Squared Distortion of Random Referee by 21. This is in contrast to the $\Omega(|\Se|)$ Squared Distortion of any favorite only mechanism (see Theorem~\ref{theorem:squaredDIstortionLowerBound}).

\begin{theorem}
\label{theorem:squaredDistortionRandomBargain}
The Squared Distortion of Random Referee is at most 21.
\end{theorem}
\begin{proof}
	Let $OPT$ be the optimal squared social cost, that is, the squared social cost of $a^*$, where $a^* \in \Se$ is the social cost minimizer. Recall that $Z_u = d(u, a^*)$. Then $$OPT = \left( \sum_{u \in N} Z_u \right)^2 = \sum_{u \in N} \sum_{v \in N} Z_u Z_v.$$ Let $ALG$ be the expected squared social cost of Random Referee. The expectation can be written out as $$ALG = \frac{1}{|N|^3} \sum_{u, v, w \in N} \left( \sum_{x \in N} d(C(\{u,v\}, w), x) \right)^2 .$$ Let $\alpha_{uvw} = \min \left( \max(Z_u, Z_v), Z_w + \min(Z_u, Z_v) \right).$ We apply Lemma~\ref{lemma:bargainTechnical} and simplify. 
\begin{equation*}
	ALG  \leq \frac{1}{|N|^3} \sum_{u, v, w \in N} \left( \sum_{x \in N} (Z_x + 2\alpha_{uvw}) \right)^2
	%& \leq \frac{1}{|N|^3} \sum_{u, v, w \in N} \left( 4 \alpha_{uvw}^2 |N|^2 + \left(\sum_{x \in N} Z_x \right)^2 + 4 \alpha_{uvw} |N| \sum_{x \in N} Z_x \right) \\
\end{equation*}
Noting that $\alpha_{uvw}$ does not depend on $x$, we can expand the square and simplify to find
\begin{equation*} 	
	ALG \leq OPT + \frac{4}{|N|^3} \sum_{u, v, w \in N} \left( \alpha_{uvw}^2 |N|^2 +  \alpha_{uvw} |N| \sum_{x \in N} Z_x \right). 
\end{equation*}
Now we sum each term separately. Let 
\begin{equation*}
\begin{split}
	& T_1 = \frac{4}{|N|^3} \sum_{u,v,w \in N} \alpha_{uvw}^2 |N|^2 \\
	& T_2 = \frac{4}{|N|^3} \sum_{u,v,w \in N} \alpha_{uvw} |N| \sum_{x \in N} Z_x
\end{split}
\end{equation*}
We will use the following basic facts: for any real numbers $a, b \geq 0, (\min(a,b))^2 \leq a \cdot b$, $\max(a,b) \leq a+b$, and $\min(a,b) \cdot \max(a,b) = a \cdot b$. 
\begin{equation*}
\begin{split}
	T_1 & \leq \frac{4}{|N|} \sum_{u, v, w \in N} \left( \max(Z_u, Z_v\right) \left( Z_w + \min(Z_u, Z_v) \right) \\
	 & \leq \frac{4}{|N|} \sum_{u, v, w \in N} \left( (Z_u + Z_v ) Z_w + Z_u Z_v \right) \\
%	 & = \frac{4}{|N|} \sum_{u, v, w \in N} \left( Z_u Z_w + Z_v Z_w + Z_u Z_v \right) \\
	  & = 12 \, OPT
\end{split}
\end{equation*}

Similarly, we analyze the second term using the fact that $\alpha_{uvw} \leq Z_u + Z_v$.
\begin{equation*}
\begin{split}
	T_2 & = \frac{4}{|N|^3} \sum_{u, v, w \in N} \alpha_{uvw} |N| \sum_{x \in N} Z_x \\
	& \leq \frac{4}{|N|^2} \sum_{u, v, w \in N} \left( Z_u + Z_v \right) \sum_{x \in N} Z_x \\
%	& \leq \frac{4}{|N|} \sum_{u, v, x \in N} \left( Z_x Z_u + Z_x Z_v \right) \\
	& = 8 \, OPT
\end{split}
\end{equation*}

Adding together all of the terms, $ALG \leq 21 \, OPT$. 
\end{proof}

This immediately yields the root mean square Distortion bound via Jensen's inequality. 

\begin{corollary}
	The Distortion of Random Referee is at most $\sqrt{21} \approx 4.583$.
\end{corollary}

%%%%%%%%%%%%%%%%%%%

\section{Distortion of Random Referee on the Euclidean Plane}
\label{sec:euclidean}
%The upper bound of 4.583 on the Distortion of Random Referee in general metric spaces may seem unsatisfying in the sense that it is worse than the arguably more naive mechanisms of Random Dictatorship. One natural question to ask is whether a more subtle analysis of Random Referee for a particular metric space of interest might reveal an upper bound on Distortion that strictly beats Random Dictatorship for arbitrary $|\S|$. 
Though the upper bound on the Distortion of Random Referee is slightly worse than that of Random Dictatorship, we now show another advantage of using a comparison query: Such mechanisms can exploit structure in specific metric spaces that favorite-only mechanisms cannot. %We now analyze the special case of the Euclidean plane and prove that the Distortion of Random Referee indeed beats that of any top-$k$ only mechanism.
%In this section, we answer this question in the affirmative by analyzing Random Referee on the Euclidean plane under the simplifying assumption that $p_u \in \F$ for all $u \in N$.
%In Section~\ref{sec:euclidean}, we present a second reason why comparison queries are useful.  
In other words, we show that the Distortion of Random Referee improves significantly for more structured metric spaces, while top-$k$ only mechanisms do not share this property. 

We examine the Distortion of Random Referee on a specific canonical metric of interest: the Euclidean plane when $d(u, p_u) = 0$ for every $u \in N$. The second assumption, functionally equivalent to assuming $N \subseteq \Se$, simplifies our analysis considerably, and corresponds to the 0-decisive case in~\cite{metricSocialChoiceRandomDictator,2-agree}. We note that in examples where $|\Se|$ is very large, the assumption becomes more innocuous. If we consider our opening example of public budgets, the assumption is something like this: every agent is allowed to propose their absolute favorite over all budgets, and we assume that this budget $p_u$ has dis-utility of 0. We consider the Euclidean plane because the problem of minimizing distortion on the real line can be solved exactly~\cite{metricSocialChoiceRandomDictator}, whereas we are unaware of any results for the Euclidean plane that are stronger than those for general metric spaces. 

% Note that the 1 dimensional case of the real line can be solved exactly~\cite{deterministic}.  

%Even in this restricted setting, we first prove in Theorem~\ref{theorem:lowerBoundGeneral} that any randomized social choice mechanism has Distortion lower bounded by 1.2, even when all bliss points are feasible, and 1.11, even when the alternatives are additionally restricted to the Euclidean plane. We further show in Theorem~\ref{theorem:lowerBoundDistortionEuclidean} that any top-$k$ only mechanism has Distortion lower bounded by 2 as $|\S|$ and $|N|$ become large. 

%On the positive side, we prove in Theorem~\ref{theorem:ordinalrandomBargain} that Random Referee has Distortion at most $1.97$. While the gap is quantitatively small, it is qualitatively interesting: we demonstrate that by using a {\em single} comparison query, Random Referee can exploit the structure of the metric space to improve Distortion, whereas Random Dictatorship or any other top-$k$ only mechanism cannot exploit such structure in any reasonable fashion. We conjecture that our result about Random Referee beating any top-$k$ mechanism holds for Euclidean spaces in any dimension and without the bliss point assumption, and present some evidence to support this conjecture.

\medskip
\noindent \textbf{Lower Bounds for Distortion in the Restricted Model.}
We begin by giving lower bounds to demonstrate that our simplifying assumptions still result in a nontrivial problem in two senses: (1) the Distortion of randomized social choice mechanisms are still bounded away from 1 and (2) any top-$k$ only mechanism (one that elicits the top $k$ preferred alternatives of agents) for constant $k$ has Distortion at least 2 as $|N|$ and $|\Se|$ become large. The proofs of the lower bounds are in the appendix.

\begin{theorem}
\label{theorem:lowerBoundGeneral}
    The Distortion of a randomized social choice mechanism is at least 1.2 generally, and at least 1.118 for the Euclidean plane, even when $d(u, p_u) = 0$ for every $u \in N$.
\end{theorem}

\begin{theorem}
\label{theorem:lowerBoundDistortionEuclidean}
	The Distortion of any top-$k$ only mechanism goes to $2$ as $|\Se|$ and $|N|$ become large, even on the Euclidean plane when $d(u, p_u) = 0$ for every $u \in N$.
\end{theorem}

\medskip
\noindent \textbf{Upper Bound for Distortion of Random Referee in the Restricted Model.}
Our positive result in this section demonstrates that a single comparison query is sufficient to construct a mechanism (Random Referee) with Distortion bounded below $1.97$ for arbitrary $|\Se|$ and $|N|$. We note that the bound in the theorem seems very slack; our goal is simply to show that using a comparison query provably decreases Distortion. We conjecture the actual bound is below 1.75 based on computer assisted search, but leave proving this stronger bound as an interesting open question.

\begin{theorem}
\label{theorem:ordinalrandomBargain}
	The worst case distortion of Random Referee is less than 1.97 when $d(u, p_u) = 0$ for every $u \in N$, $\mathcal{S} \subseteq \mathbb{R}^2$, and for $x,y \in \mathbb{R}^2$ $d(x,y) = \|x-y\|_2$.
\end{theorem}

In the remainder of this section, we sketch the proof of Theorem~\ref{theorem:ordinalrandomBargain}. Proofs of some of the technical lemmas are in the appendix. Suppose for a contradiction that there is a set $N$ of agents in $\mathbb{R}^2$ such that the Distortion of Random Referee is at least $1.97$ under the Euclidean metric. We will successively refine this hypothesis for a contradiction, finally arguing that it implies that some ``bad'' instance would appear on an exhaustive computer assisted search over a finite grid.

The crucial lemmas bound the \textit{pessimistic distortion} of any set of five points in $\mathbb{R}^2$ and relate this to the actual distortion. The quantity is pessimistic because it allows ``OPT'' to choose a separate point for every 5-tuple; the numerator of each 5-tuple is just a rewriting of the expected social cost of Random Referee. In this section, since we assume $d(u, p_u) = 0$, it will not be necessary to refer to $u$ and $p_u$ separately, and it will be convenient to let $C(\{p_u,p_v\},p_w) = \mbox{argmin}_{a \in \{p_u, p_v\}} d(p_w, a)$.   

\begin{definition}
\label{definition:pessimisticDistortion}
    The \textbf{pessimistic distortion}\footnote{Note that pessimistic distortion is technically a function of the mechanism used; we will use it exclusively in reference to Random Referee.} of $\mathcal{P} = \{ x_1,x_2,x_3,x_4,x_5 \} \subset \mathbb{R}^2$ is defined as $$PD(\mathcal{P}) := \frac{SCRR_{avg}(\mathcal{P})}{OPT_{avg}(\mathcal{P})}$$ where $SCRR_{avg}(\mathcal{P})$ is the average social cost of Random Referee, which we can write as $$\frac{1}{30} \sum_{i=1}^5 \, \sum_{j>i}^5 \sum_{k \neq i,j} \; \sum_{l \neq i,j,k} d(x_l, C(\{x_i, x_j\}, x_k)),$$ and the average cost of the optimal solution is $$OPT_{avg}(\mathcal{P}) := \min_{y \in \mathbb{R}^2} \frac{1}{5} \sum_{r=1}^5 d(x_r, y).$$ Finally, $PD(\mathcal{P}) = 1$ if $x_1 = x_2 = x_3 = x_4 = x_5$.
\end{definition}

%--------------------------------------------------------%

The first lemma relates this pessimistic distortion to the actual Distortion of Random Referee. The worst case (over 5-tuples) pessimistic distortion upper bounds the Distortion of Random Referee. Interestingly, this statement is not specific to the Euclidean plane, suggesting that our approach may be broadly applicable for proving stronger Distortion bounds on other specific metrics of interest. %The argument is based on reordering summations and sampling: many of the possible random draws of agents display symmetries that we can avoid. This is important for the efficiency of the brute force computer search we will employ at a later step.
\begin{lemma}
\label{lemma:distortion}
    If $PD(P) \leq \beta$ for all $\mathcal{P}=\{x_1, \hdots, x_5\} \subset \mathbb{R}^2$ then the Distortion of Random Referee is at most $\beta$ on $\mathbb{R}^2$.
\end{lemma}
\begin{proof}
Observe that we can rewrite the Distortion of Random Referee as a summation over all possible 5-tuples of points. Let $\mathcal{P}$ be a multiset of points (i.e., possibly with repeats). Let $\rho(\mathcal{P})$ be an ordering of $\mathcal{P}$. Let $$C(\rho(\mathcal{P})) = C(\{\rho_1(\mathcal{P}), \rho_2(\mathcal{P})\}, \rho_3(\mathcal{P})).$$ We can rewrite the Distortion of Random Referee over permutations of 5-tuples as
%Let $$g(\mathcal{P}) = \frac{1}{5!} \sum_{\rho(\mathcal{P})} d(\rho(\mathcal{P})_4, C(\{\rho(\mathcal{P})_1, \rho(\mathcal{P})_2\}, \rho(\mathcal{P})_3)) + d(\rho(\mathcal{P})_5, C(\{\rho(\mathcal{P})_1, \rho(\mathcal{P})_2\}, \rho(\mathcal{P})_3)).$$  
\begin{equation*}
	\frac{\frac{1}{|N|^4} \sum_{i,j,k,l \in N} d(p_l, C(\{p_i, p_j\}, p_k)}{\frac{1}{|N|} \sum_{i \in N} d(p_i, a^*)} 
\end{equation*}
\begin{equation*}
	= \frac{\sum_{\mathcal{P} \subset \mathcal{P}_N} \Pr(\mathcal{P}) \frac{1}{5!} \sum_{\rho(\mathcal{P})} \sum_{l \in \{4,5\}} d(\rho_l(\mathcal{P}), C(\rho(\mathcal{P})))}{\frac{1}{|N|} \sum_{i \in N} d(p_i, a^*)}. 
\end{equation*}
In words, we are considering all 5-tuples (with replacement) of agent points, and for each we consider the average distance over all orderings of the five agent points of the distance between the last two points and the outcome of Random Referee when the first three agents participate. This is in turn upper bounded by allowing OPT to choose a different $a^*$ for every 5-tuple, so that the Distortion is at most
\begin{equation*}
\begin{split}
	&  \frac{\sum_{\mathcal{P} \subset P_N} \Pr(\mathcal{P})\frac{1}{5!} \sum_{\rho(\mathcal{P})} \sum_{l \in \{4,5\}} d(\rho_l(\mathcal{P}), C(\rho(\mathcal{P})))}{\sum_{\mathcal{P} \subset P_N} \Pr(\mathcal{P}) OPT_{avg}(\mathcal{P})} \\
	& \leq \max_{\mathcal{P} \subset P_N} \frac{\frac{1}{5!} \sum_{\rho(\mathcal{P})} \sum_{l \in \{4,5\}} d(\rho_l(\mathcal{P}), C(\rho(\mathcal{P})))}{OPT_{avg}(\mathcal{P})}.  
\end{split}
\end{equation*}

To complete the proof, note that $SCRR_{avg}(\mathcal{P})$ is equal to the numerator. We start with all 120 orderings of the $5$ points and avoid double counting the symmetric cases that arise from swapping the two points from which we take the argmin and swapping the two points from which we measure distance.
\end{proof}

Now we can refine our original hypothesis: without loss of generality, assume for a contradiction that there is a multiset $\mathcal{P}=\{x_1, \hdots, x_5\} \subseteq \mathbb{R}^2$ with $PD(\mathcal{P}) \geq 1.97$.

\begin{lemma}
\label{lemma:pessimistic}
    $PD(\mathcal{P}) < 1.97 $.
\end{lemma}

Lemma~\ref{lemma:pessimistic} provides the contradiction to our hypothesis and establishes Theorem~\ref{theorem:ordinalrandomBargain}. We outline the main ideas of the proof. First, we consider a grid on the Euclidean plane, and argue that we can assume a certain canonical structure of $\mathcal{P}$ in relation to that grid by scaling, translating, and rotating. Next, we carefully argue for how much the pessimistic distortion would change if every point in $\mathcal{P}$ were snapped to a grid point. We show that if there is some $\mathcal{P}$ with pessimistic distortion at least 1.97, there must be a 5-tuple of points on a sufficiently fine (but finite) grid with a sufficiently large pessimistic distortion. However, we employ computer assisted analysis to brute force search over such a grid, and find no such bad example. 

\medskip
\noindent \textbf{Discussion.} Our analysis is slack in two ways: (1) we have to interpolate between grid points, and (2) we consider pessimistic distortion over 5-tuples. Both are computational constraints: (1) because we cannot simulate an arbitrarily fine grid, and (2) because there is a combinatorial blow up in the search space when considering larger tuples (note that by considering 5-tuples, we implicitly allow OPT to choose a separate optimal solution for each 5-tuple). For Random Referee, the worst case example found by computer simulations for grid points is fairly simple: The 5 points lie on a straight line with pessimistic distortion 1.75. We conjecture the same example is the worst case even in the continuous plane, which suggests a distortion bound of at most 1.75. We leave finding the exact bound as an open question, as our result is sufficient to demonstrate that comparison queries can take advantage of structure that top-$k$ queries cannot. 
%We conjecture the actual Distortion bound for Random Referee is below $1.6$, and nailing it precisely is an interesting open question. We believe that mechanisms such as Copeland that use more comparison information will only do better on Euclidean spaces, and we leave finding the optimal Distortion mechanism as another open question.

%\subsection{Discussion}

%Further, the above technique generalizes to higher dimensional Euclidean space. Given that the above proof uses only 5-tuples of points regardless of dimension of the space, without loss of generality, we can assume these lie in 4 dimensions since the space is Euclidean. Computer search over a coarse grid again shows that the worst case is achieved by the two-dimensional example of five points on a straight line, which means the Distortion bound is better than $2$. However, we have not been able to run the search on a fine enough grid to prove this claim formally. We leave these as interesting open questions. 

%--------------------------------------------------%

\section{Favorite Only Mechanisms: Random Oligarchy}
\label{sec:oligarchy}
Recall that a favorite query asks an agent $u \in N$ for her favorite alternative $p_u$. In this section, we return to the general model (arbitrary metrics and not assuming that $d(u, p_u) = 0$) and study mechanisms that are restricted to only use favorite queries. We show that essentially optimal Distortion as a function of $|\Se|$ is achieved by a simple mechanism that uses just $3$ queries. We call this mechanism Random Oligarchy.

\begin{definition}
\label{def:RO}
	The \textbf{Random Oligarchy (RO)} mechanism samples three agents $u,v,w \in N$ independently and uniformly at random with replacement. All three are asked for their favorite alternatives $p_u$, $p_v$, and $p_w$ in $\Se$. If the same alternative is reported at least twice, output that, else output one of the three alternatives uniformly at random.
\end{definition}

We prove that Random Oligarchy has the best of both worlds with respect to the other favorite only mechanisms of Random Dictatorship and 2-Agree. Unlike 2-Agree, Random Oligarchy has constant sample complexity and the same Distortion bound of 3 for large $|\Se|$ as Random Dictatorship. However, like 2-Agree, it outperforms Random Dictatorship for small $|\Se|$. The proof of the theorem is given in the appendix. %In fact, comparing against the lower bound for any favorite only mechanism from~\cite{2-agree} allows us to see in Figure~\ref{figure:oligarchy} that Random Oligarchy is essentially optimal among all favorite only mechanisms.

%We analyze Random Oligarchy using a technical lemma proven in~\cite{2-agree}, and achieve Distortion bounds for small $|\S|$ that are comparable to those from the 2-Agree mechanism introduced in the same work. In particular, both 2-Agree and Random Oligarchy have lower Distortion than the natural baseline algorithm of Random Dictatorship when $|\S|$ is small (Random Dictatorship has Distortion upper and lower bounded by 2 for arbitrary $|\S| \geq 2$, assuming $|N|$ can be arbitrarily large). However, Random Oligarchy has two key advantages over 2-Agree. The first is that it only requires $3$ favorite queries as opposed to 2-Agree, which may require $|\S|+1$ favorite queries. The second advantage is that unlike 2-Agree, Random Oligarchy is globally bounded to have Distortion at most 2 for arbitrary $|\S|$, and thus establishes a ``best of both worlds'' style of guarantee with respect to the 2-Agree and Random Dictatorship mechanisms. We first restate the necessary technical lemma.

\begin{theorem}
\label{theorem:oligarchy}
	The Distortion of Random Oligarchy is upper bounded by $3$ for arbitrary $|\Se|$, and by the following expression for particular $|\Se|$. $$1 + 2\max_{p \in [0,1]} \left( 1 + p^2(p-2) + \frac{(p-1)^3}{|\Se|-1} \right)$$
\end{theorem}

Figure~\ref{figure:oligarchy} shows the Distortion bounds of favorite only mechanisms. Comparing against the lower bound for any favorite only mechanism from~\cite{2-agree} allows us to see that Random Oligarchy is essentially optimal among all favorite only mechanisms. In comparison to existing mechanisms, Random Oligarchy outperforms Random Dictatorship for small $|\Se|$, and outperforms 2-Agree for large $|\Se|$, while only using three favorite queries.
\begin{figure}
\centering
\includegraphics{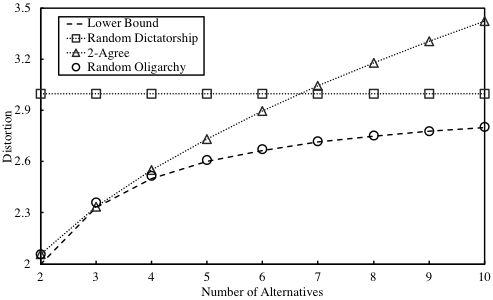} 
\iffalse \begin{tabular}{| c | c | c | c |}
\hline
 $|\Se|$ & Lower Bound & 2-Agree & Random Oligarchy \\ \hline \hline
 2 	& 2		& 2.056 	&  2.056	\\ \hline
 3 	& 2.333 	& 2.334	&  2.359	\\ \hline    
 4 	& 2.5		& 2.551	&  2.515	\\ \hline
 5 	& 2.6		& 2.734	& 2.609	\\ \hline
 6 	& 2.666	& 2.896	& 2.673	\\ \hline
 7 	& 2.714	& 3.043	& 2.719	\\ \hline
 8 	& 2.75	& 3.178	& 2.753	\\ \hline
 9	& 2.777	& 3.304	& 2.780	\\ \hline
 10	& 2.8		& 3.422	& 2.802	\\ \hline
 \hline 
\end{tabular} \fi
\caption{Distortion of Favorite Only Mechanisms. Upper bounds for Random Oligarchy are from Theorem~\ref{theorem:oligarchy}. Lower bounds and 2-Agree upper bounds are from~\cite{2-agree}. Random Dictatorship is analyzed in~\cite{metricSocialChoiceRandomDictator}.}
\label{figure:oligarchy}
\end{figure}

%In summary, if we were only concerned with simplicity and Distortion, then Random Oligarchy would appear to be the best known randomized social choice mechanism in the sense that it uses only 3 very simple favorite queries, is nearly optimal on Distortion among mechanisms of this class, and has an overall upper bound on Distortion matching the best known social choice mechanism of \textit{any} kind-Random Dictatorship. However, as we have already argued, simplicity and Distortion are not the only desiderata we have for a randomized social choice mechanism. We also wish to bound the risk involved in the randomization, which we analyze via Squared Distortion. Unfortunately, Random Oligarchy (like Random Dictatorship) has large Squared Distortion in general.

\section{Open Directions}
\label{sec:open}
%The analysis of Distortion, sample complexity, and favorite only mechanisms in Section~\ref{sec:oligarchy} is essentially tight. However, 
In this paper, we have considered constant sample complexity mechanisms. At a high level, we hope that our work inspires future research on lightweight mechanisms for social choice in large decision spaces. We also mention some natural technical questions raised by our work.

Compared to Distortion, there is much less understanding of the Squared Distortion of mechanisms. The only universal lower bound for Squared Distortion we are aware of is 4, a consequence of the lower bound of 2 for Distortion~\cite{metricSocialChoiceRandomDictator}. Our Random Referee mechanism achieves Squared Distortion of at most 21 using just three ordinal queries. Closing this gap remains an interesting question even for mechanisms that elicit full ordinal information.%, in addition to constant sample complexity mechanisms. %Furthermore, strategic concerns remain an interesting question for our mechanisms.

The analysis in Section~\ref{sec:euclidean} of Random Referee generalizes to higher dimensional Euclidean space, still only reasoning about 5-tuples of points. Computer search over a coarse grid in four dimensions again shows that the pessimistic distortion bound is better than $2$. However, we have not been able to run the search on a fine enough grid to prove a Distortion bound formally. We leave proving this for higher dimensional Euclidean spaces as an interesting open question. Additionally, we hope that related methods may be of general interest for proving tighter Distortion bounds of mechanisms on restricted metric spaces. 

It remains unclear whether using a constant number of queries greater than 3 would meaningfully improve our results. E.g., consider the natural extension of Random Referee: sample the favorite points of $k$ agents and ask a random referee to choose their favorite from among these. Using $k > 2$ does not straightforwardly decrease the Squared Distortion bound of 21 for $k$ = 2. Also, as $k$ becomes large, this mechanism devolves to Random Dictatorship. Another natural question is whether $O(k)$ comparison queries are necessary/sufficient to bound the $k$'th moment of Distortion. We leave these as additional open questions. 

%The proof in that section uses only 5-tuples of points regardless of dimension of the space, therefore, without loss of generality, we can assume these lie in 4 dimensions since the space is Euclidean.  

%on the Euclidean plane can easily be extended to higher dimensions, and still only requires reasoning over 5-tuples of points, which lie in $4$ dimensions without loss of gnerality. Searches over a coarse grid do not reveal any worse examples than in two dimensions. This leads us to conjecture that the same argument may go through to show that Random Choice has Distortion less than two in the bliss point feasible model for arbitrary dimensions under the Euclidean norm.

\section*{Acknowledgments}
Brandon Fain is supported by NSF grants CCF-1637397 and IIS-1447554. Ashish Goel is supported by NSF grant CCF-1637418 and ONR grant N00014-15-1-2786. Kamesh Munagala is supported by NSF grants CCF-1408784, CCF-1637397, and IIS-1447554. Much of this work was done while Nina Prabhu was a student at The North Carolina School of Science and Mathematics, in cooperation with Duke University.

%\clearpage
\bibliographystyle{abbrv}
\bibliography{refs}

\section*{Appendix}
\label{sec:appendix}
\subsection*{Proof of Theorem~\ref{theorem:squaredDIstortionLowerBound}}
We want to construct a single ordinal profile over top-$k$ preferences such that for all randomized social choice mechanisms, there is some instantiation of these top-$k$ preferences in a metric space on which the mechanism has Squared Distortion $\Omega(|\Se|)$. The top-$k$ ordinal profile is simple: Each agent $u \in N$ has completely unique top-$k$ preferred feasible alternatives $S_u \subseteq \Se$. There are no other alternatives, so $|\Se| = k|N| = \Theta(|N|)$ (note that a top-$k$ only mechanism must have constant $k$).  

Any randomized social choice mechanism must choose a distribution on $\Se$, so in particular, there will be some $u^* \in N$ such that the mechanism puts at least $1/|N|$ probability mass on the top-$k$ preferred alternatives of $u^*$. Given this, consider a metric space, consistent with the top-$k$ preference profile, where $S_u^*$ forms a small clique well separated from all of the other alternatives which form a large clique. More formally, let $0 < \epsilon \leq 1/|N|$. All pairwise distances are $\epsilon$ except those between the small and large clique. Instead, $d(u, a) = 1$ for all $u \neq u^*, a \in S_{u^*}$, and similarly $d(u^*, a) = 1$ for all $a \in \Se \backslash S_{u^*}$.

Then the Squared social cost of choosing $a \in S_{u^*}$ is simply $(|N|-1)^2$, whereas the optimal solution chooses any $a \in \Se \backslash S_{u^*}$ for Squared social cost of $\left( 1 + \epsilon (|N|-1)\right)^2 \leq \left( 1+\epsilon |N|\right)^2$. Since the mechanism chooses $a \in S_{u^*}$ with probability at least $1/|N|$, we can lower bound the Squared Distortion of any top-$k$ only mechanism $f$ as follows.
\begin{equation*}
\begin{split}
	Distortion^2 \left( f \right) & \geq \frac{1}{|N|} \frac{\left(|N|-1\right)^2}{\left( 1+\epsilon |N|\right)^2} \\
	& \geq \frac{\left(|N|-1\right)^2}{4|N|} = \Omega(|\Se|) 
\end{split}
\end{equation*}    

%-----------------------------------------------------%
%-----------------------------------------------------%

\subsection*{Proof of Theorem~\ref{theorem:lowerBoundGeneral}}
\begin{figure}[h]
\centering
\begin{tikzpicture} 
\fill[black] (1,1) circle (0.1);
\draw[black, thick] (1,1) -- (0,0);
\draw[black, thick] (1,1) -- (1,0);
\draw[black, thick] (1,1) -- (2,0);
\fill[gray] (2,1.75) circle (0.0);
\draw [dashed] plot [smooth, tension = 1] coordinates {(0,0) (0.5,-0.25) (1,0)};
\draw [dashed] plot [smooth, tension = 1] coordinates {(1,0) (1.5,-0.25) (2,0)};
\draw [dashed] plot [smooth, tension = 1] coordinates {(0,0) (1,-1) (2,0)};
\fill[red] (-0.1,-0.1) rectangle (0.1,0.1);
\fill[red] (0.9,-0.1) rectangle (1.1,0.1);
\fill[red] (1.9,-0.1) rectangle (2.1, 0.1);
\end{tikzpicture}
\caption{Example for Theorem~\ref{theorem:lowerBoundGeneral}. The red squares are agents as well as alternatives.}
\label{figure:lowerBoundExample}
\end{figure}
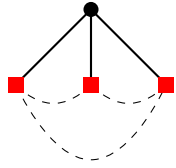
    We begin with the more general claim. Consider a star graph with three leaves and three agents, one agent at each leaf, as depicted in Figure~\ref{figure:lowerBoundExample}. All four points are also alternatives, and the distance between an agent and the colocated alternative is 0. Call the central vertex $B$ and the set of outer vertices $R$. We want to construct a single ordinal preference profile such that for all randomized social choice mechanisms, there is some metric space that induces this preference profile and for which the mechanism has Distortion at least 1.2.  The preference profile is simple: each of the three agents prefer their colocated alternative to the central alternative to either of the other leaf alternatives. We induce this profile via two metric spaces. In both spaces, $\forall v \in R, \; d(v,B) = 1$ (that is, the solid lines in Figure~\ref{figure:lowerBoundExample}); the spaces will differ in the distances between points \textit{inside} of $R$ (that is, the dashed lines in Figure~\ref{figure:lowerBoundExample}).
    
    First, suppose that $\forall v, v' \in R, \; v \neq v', \; d(v, v') = 2$.  Then the social cost of $B$ is 3, and the social cost of every vertex in $R$ is $4$; clearly the distortion of $B$ is 1 and the distortion of $v \in R$ is $4/3$. Let $\Pr(R)$ be the probability that a randomized ordinal social mechanism chooses some $v \in R$, and let $\Pr(B)$ be the probability that the same mechanism chooses $B$.  Then on this metric, the Distortion of the mechanism is (in expectation) just $\Pr(B) + (4/3)\Pr(R)$.
    
    Second, suppose that $\forall v, v' \in R, \; v \neq v', \; d(v, v') = 1+\epsilon$ for some $\epsilon \in (0,0.5)$.  Then the social cost of $B$ is 3, and the social cost of every alternative in $R$ is $2(1+\epsilon)$, so the distortion of $B$ is $\frac{3}{2(1+\epsilon)}$ and the distortion of $v \in R$ is $1$. Then on this metric, the Distortion of the mechanism is (in expectation) $\frac{3}{2(1+\epsilon)} \Pr(B) + \Pr(R)$.
    
    A social choice mechanism with only ordinal information cannot distinguish the cases, and must set some values for $\Pr(R)$ and $\Pr(B)$ that sum to 1, so such a mechanism cannot have Distortion less than $$\max \left( \Pr(B) + \frac{4}{3}\Pr(R), \; \frac{3}{2(1+\epsilon)} \Pr(B) + \Pr(R) \right ).$$ This expression is minimized by setting $\Pr(R)$ to $\frac{2+2\epsilon}{5-4\epsilon}$ (and $\Pr(B) = 1-\Pr(R)$) yielding a Distortion of $$\frac{4}{3} - \left(\frac{1}{3}\right) \left( \frac{2+2\epsilon}{5-4\epsilon} \right) \longrightarrow 6/5 \mbox{ as } \epsilon \rightarrow 0 $$

The argument for the lower bound of $1.118$ on the Euclidean plane is similar. Now, the first case places the three points in $R$ on a unit circle around the central point $B$, equidistant from one another. The law of cosines implies that $\forall v, v' \in R, \; v \neq v', \; d(v, v') = \sqrt{2 - 2\cos(120^{\circ})} \approx 1.732$. In the second case, the angular distance between consecutive points in $R$ is $60^{\circ}$, so that two pairwise distances are $1$, and one pairwise distance is  $\sqrt{2 - 2\cos(120^{\circ})}$. Completing the same argument as before yields the bound of 1.11.

%-----------------------------------------------------%
%-----------------------------------------------------%

\subsection*{Proof of Theorem~\ref{theorem:lowerBoundDistortionEuclidean}}
The construction for this lower bound follows that of Theorem~\ref{theorem:squaredDIstortionLowerBound}. We briefly consider the argument for completeness. There are $|\Se|/k$ agents, and for each agent $u$, there are $k$ unique alternatives $S_u \subset S$ that are $\epsilon$ distance away from them. $S_u$ are the top-$k$ preferred alternatives for $u$. There are no other alternatives. Consider any top-$k$ only mechanism (recall that $k$ is constant). Such a mechanism chooses a distribution over alternatives, and thus $\exists u^* \in N$ with top-$k$ preferred alternatives $S_{u^*}$ such that the mechanism puts at least $1/|N|$ probability mass on $S_{u^*}$. 

Suppose that the other $|N|-1$ agents are arrayed in a circle such that the pairwise distance of diameter $\delta$ (clearly, if $\epsilon$ is small enough, such an arrangement need not violate our earlier construction). Further, suppose the distance from $u^*$ to the center of this circle is 1. Then the social cost of any alternative in $S_{u^*}$ is at least $(|N|-1)(1-\epsilon-\delta/2)$, and the social cost of any alternative not in $S_{u^*}$ is at most $1+\epsilon \delta/2+ (|N|-2)(\delta + 2\epsilon)$. Therefore, the Distortion of choosing any alternative in $S_{u^*}$ is at least $$\frac{(|N|-1)(1-\epsilon-\delta/2)}{1+\epsilon \delta/2+ (|N|-2)(\delta + 2\epsilon)}.$$ 

As $\epsilon$ and $\delta$ approach 0, the above approaches $|N|-1$. The Distortion of choosing any alternative not in $S_{u^*}$ is clearly just $1$. Since the mechanism chooses an alternative in $S_{u^*}$ with probability at least $1/|N|$ (say $1/|N| + \alpha$, where $\alpha \geq 0$) and an alternative not in $S_{u^*}$ with the remaining $1 - 1/|N|-\alpha$ probability, the over Distortion is lower bounded by 
\begin{equation*}
\begin{split}
	& \left( \frac{1}{|N|} + \alpha \right) (|N|-1) + \left( \frac{|N|-1}{|N|} - \alpha \right) \\
	& =  2 - 2/|N| + \alpha \left( |N| - 2 \right) \\
	& \geq 2 - 2/|N| \\ 
	& \rightarrow 2 \mbox{ as } |N| \rightarrow \infty.
\end{split}
\end{equation*}

%-----------------------------------------------------%
%-----------------------------------------------------%

\subsection*{Proof of Lemma~\ref{lemma:pessimistic}}

To begin, we argue that we can assume a canonical structure to $\mathcal{P}$ in relation to a finite grid.

\begin{lemma}
\label{lemma:structure}
    Let $G$ be a $\delta$-fine grid: $G = \{0, \delta, 2\delta, \hdots, \frac{\delta-1}{\delta}, 1\}^2$ where $\delta > 0$ and $1/\delta \in \mathbb{N}$. There is a set of points $\mathcal{Q} = \{y_1, y_2, y_3, y_4, y_5\}$ such that $PD(\mathcal{Q}) = PD(\mathcal{P})$ and $\mathcal{Q}$ has the properties:
    \begin{enumerate}
        \item $ \mathcal{Q} \subset [0,1+\delta]^2 $,
        \item $ \max_{i,j} d(y_i,y_j) = 1 $,
        \item and these two maximally separated points $y_i$ and $y_j$ have the forms $y_i = (\alpha \delta, 0), y_j = (\alpha \delta, 1)$ for some $\alpha \in \{0,1,2,\hdots, 1/\delta\}$
    \end{enumerate}
\end{lemma}

\begin{proof}
We note that the pessimistic distortion of $\mathcal{P}$ is invariant to scaling, rotation, and translation of $\mathcal{P}$. This follows because the pessimistic distortion is still defined in terms of euclidean distances. Given this, the argument for the Lemma is simply that we can construct $\mathcal{Q}$ by rotations, translations, and scaling.
    
First scale $\mathcal{P}$ so that the maximum distance separating two points is 1 (note that at least two points are not equal). Next, rotate the points so that the line between these maximally separated points is vertical. Finally, translate the points until these maximally separated points have the appropriate forms - that is to say, until they lie exactly at grid points, and all other points are within $[0, 1+\delta]^2$. To see that this is possible, note that since they are the maximally separated points, no other points can lie outside $[0,1]$ in the vertical dimension. In the horizontal dimension, the total width spanned by the other points is at most 1, and we may need to expand by as much as $\delta$ in order to align the maximally separated points with grid points.
\end{proof}

The introduction of a $\delta$-fine grid anticipates the computer analysis we employ. Lemma~\ref{lemma:structure} allows us to further refine our assumption for a contradiction: Suppose without loss of generality\footnote{For simplicity of notation, we will assume $\mathcal{P}$ lies in a $[0,1]$ rather than $[0,1+\delta]$, and will simply run our computer analysis for the expanded grid. Similarly, relabeling $x_1$ and $x_2$ as the maximally separated points is just a notational convenience.} that there is some $\mathcal{P}=\{x_1, \hdots, x_5\} \subseteq [0,1]^2$ such that $PD(\mathcal{P}) \geq \alpha$, $d(x_1, x_2) = 1$, and $x_1, x_2$ are of the form $(\alpha \delta, 0), y_j = (\alpha \delta, 1)$ for some $\alpha \in \{0,1,2,\hdots, 1/\delta\}$.

\medskip
Define $\phi:[0,1]^2 \rightarrow G$ as $\phi(x) = \mbox{argmin}_{v \in G} d(x,v)$. In words, $\phi$ maps a general point in $[0,1]$ to its nearest point on our $\delta$-fine grid $G$. We will argue (roughly) that the pessimistic distortion of $\phi(\mathcal{P}):= \{\phi(x_1), \phi(x_2) \hdots, \phi(x_5)\}$ closely approximates that of $\mathcal{P}$, where the pessimistic distortion of of $\phi(\mathcal{P})$ is something we can compute directly in a brute force computer search. Recall Definition~\ref{definition:pessimisticDistortion}. We begin by bounding the numerator.
%    \begin{equation*}
%        \frac{PD(\mathcal{P})}{PD(\phi(\mathcal{P}))} = \left(\frac{SCRC_{avg}(\mathcal{P})}{SCRC_{avg}(\phi(\mathcal{P}))}\right) \left(\frac{OPT_{avg}(\phi(\mathcal{P}))}{OPT_{avg}(\mathcal{P})}\right).
%    \end{equation*}
%We bound the ratios separately. 

Let $ C_{\phi}(\{x,y\},z) = C(\{\phi(x),\phi(y)\},\phi(z))$. Clearly we have from the definition of $\phi$ that $d(x, \phi(x)) \leq \left(\sqrt{2}/2\right)\delta$ for all $x$ (and $\phi(x_1)=x_1$ and $\phi(x_2)=x_2$). However, it is not necessarily true that $d(C(\{x,y\},z), C_{\phi}(\{x,y\},z))$ is also small for all $x, y, z$. It is possible that $z$ is slightly closer to $x$ than $y$, but $\phi(z)$ is slightly closer to $\phi(y)$ than $\phi(x)$. We call such configurations indifferences.
    \begin{definition}
        Call $z$ \textbf{indifferent} with respect to $x$ and $y$ if $\left|d(z,x) - d(z,y)\right| \leq \frac{3\sqrt{2}}{2}\delta$.     
    \end{definition}
    Then we have the following fact: If $z$ is not indifferent with respect to $x$ and $y$, then $$d(C(\{x,y\},z), C_{\phi}(\{x,y\},z)) \leq \frac{\sqrt{2}}{2}\delta.$$ This follows from observing that if $z$ is not indifferent, $\phi(C(\{x,y\},z)) = C_{\phi}(\{x,y\},z))$. Now, for our point set $\mathcal{P}$, we need to bound $d(x_l, C(\{x_i, x_j\}, x_k))$. We first consider the case where $x_k$ is \textit{not} indifferent with respect to $x_i$ and $x_j$. Then we have can upper bound $d(x_l, C(\{x_i, x_j\}, x_k))$ by
    \begin{equation*}
        \begin{split}
            & d(x_l, \phi(x_l)) + d(\phi(x_l), C_{\phi}(\{x_i, x_j\}, x_k)) \\
            & \;\;\; + d(C_{\phi}(\{x_i, x_j\}, x_k), C(\{x_i, x_j\}, x_k)) \\ 
            & \leq \begin{cases} \left(\sqrt{2}\right) \delta + d(\phi(x_l), C_{\phi}(\{x_i, x_j\}, x_k)) & l \in \{3,4,5\} \\
            \left(\frac{\sqrt{2}}{2}\right) \delta + d(\phi(x_l), C_{\phi}(\{x_i, x_j\}, x_k)) & l \in \{1,2\}
            \end{cases}
        \end{split}
    \end{equation*}
    where the bound tightens for $x_1$ and $x_2$ because they are already at grid points. Suppose instead that $x_k$ \textit{is} indifferent with respect to $x_i$ and $x_j$. Then we can upper bound $d(x_l, C(\{x_i, x_j\}, x_k))$ by
    \begin{equation*}
        \begin{split}
            & d(x_l, \phi(x_l)) + d(\phi(x_l), \phi(C(\{x_i, x_j\}, x_k))) \\
           & \;\;\; + d(\phi(C(\{x_i, x_j\}, x_k)), C(\{x_i, x_j\}, x_k)) \\ 
       \end{split}
    \end{equation*}
which is at most
\begin{equation*}
\begin{cases} \left(\sqrt{2}\right) \delta + d(\phi(x_l), \phi(C(\{x_i, x_j\}, x_k))) & l \in \{3,4,5\} \\
            \left(\frac{\sqrt{2}}{2}\right) \delta + d(\phi(x_l), \phi(C(\{x_i, x_j\}, x_k))) & l \in \{1,2\}
            \end{cases}
\end{equation*}
Furthermore, since $C(\{x_i, x_j\}, x_k) \in \{x_i, x_j\}$, we know that $$d(\phi(x_l), \phi(C(\{x_i, x_j\}, x_k))) \leq \max_{q \in \{i,j\}} d(\phi(x_l, \phi(x_q))).$$ Call the set of indifferent cases $I$. We can average over the indifferences, and over $\{x_1, x_2\}$ and $\{x_3, x_4, x_4\}$, to upper bound $SCRR_{avg}(\mathcal{P})$ by
    \begin{equation*}
    \begin{split}
        & \left(\frac{3}{5}\right)(\sqrt{2}) \delta +\left(\frac{2}{5}\right) \left(\frac{\sqrt{2}}{2}\right)\delta \\
        & \;\;\; + \frac{1}{30} \sum_{I} \sum_{l \neq i, j, k} d(\phi(x_l), \phi(C(\{x_i, x_j\}, x_k))) \\
        & \;\;\; + \frac{1}{30} \sum_{\overline{I}} \sum_{l \neq i,j,k} d(\phi(x_l), C_{\phi}(\{x_i, x_j\}, x_k)) \\
        & \leq \frac{4\sqrt{2}}{5}\delta + \frac{1}{30} \sum_{I} \mbox{max}_{q \in \{i,j\}} \hspace{-0.1cm} \sum_{l \neq i,j,k} d(\phi(x_l), \phi(x_q)) \\
        & \;\;\; + \frac{1}{30} \sum_{\overline{I}} \sum_{l \neq i,j,k} d(\phi(x_l), C_{\phi}(\{x_i, x_j\}, x_k)) \\
    \end{split}
    \end{equation*}
    
    Call the summations in this last line $GridSum$. Then 
    \begin{equation*}
        \frac{SCRR_{avg}(\mathcal{P})}{GridSum} \leq 1 + \frac{4\sqrt{2} \delta}{5 \, GridSum}
    \end{equation*}
    $GridSum$ looks like a very involved quantity, but it is exactly what we will compute in our computer analysis. That is, when we are computing the pessimistic distortion on the grid for a set of five points, whenever we encounter an indifference, we take the worst case outcome over the two options. It is not hard to see that $GridSum \geq 1/5$ since there are least two points separated by a distance of $1$, so we arrive at equation~\ref{equation:SCBound}.
    \begin{equation}
    \label{equation:SCBound}
	 \frac{SCRR_{avg}(\mathcal{P})}{GridSum} \leq 1 + \left(4\sqrt{2}\right)\delta
    \end{equation}
    Now we need to bound $$\frac{OPT_{avg}(\phi(\mathcal{P}))}{OPT_{avg}(\mathcal{P})}.$$ Let \begin{equation*} 
    \begin{split}
    y_{\phi}^* & = \mbox{argmin}_y \frac{1}{5} \sum_{l=1}^5 d(\phi(x_l), y)\\
    y^* & = \mbox{argmin} _y \frac{1}{5} \sum_{l=1}^5 d(x_l, y). 
    \end{split}
\end{equation*}
    Since $y^*$ was a feasible choice for the minimization over $\phi(\mathcal{P})$, the triangle inequality implies that 
    \begin{equation*}
        \begin{split}
            \frac{1}{5} \sum_{l=1}^5 d(\phi(x_l), y_{\phi}^*) & \leq \frac{1}{5} \sum_{l=1}^5 \left( d(\phi(x_l), x_l) + d(x_l, y^*) \right) \\
            & \leq \left(\frac{3 \sqrt{2}}{10}\right) \delta + \frac{1}{5} \sum_{l=1}^5 d(x_l, y^*) \\
%            & \implies \frac{OPT_{avg}(\phi(\mathcal{P}))}{OPT_{avg}(\mathcal{P})} \leq 1 + \frac{\left(\frac{3 \sqrt{2}}{10}\right) \delta}{OPT_{avg}(\mathcal{P})}
        \end{split}
    \end{equation*}
    This implies that $$OPT_{avg}(\phi(\mathcal{P})) \leq \left(\frac{3 \sqrt{2}}{10}\right) \delta + OPT_{avg}(\mathcal{P}).$$
    Since we know that $d(x_1,x_2) = 1$, we have that $OPT_{avg}(\mathcal{P}) \geq 1/5$. So we get that
    \begin{equation}
    \label{equation:OPTBound}
        \frac{OPT_{avg}(\phi(\mathcal{P}))}{OPT_{avg}(\mathcal{P})} \leq 1 + \left(\frac{3 \sqrt{2}}{2}\right) \delta.
    \end{equation}    
    Combining equation~\ref{equation:SCBound} and equation~\ref{equation:OPTBound}:
    \begin{equation*}
        PD(\mathcal{P}) \frac{OPT_{avg}(\phi(\mathcal{P}))}{GridSum} \leq \left( 1 + (4\sqrt{2}) \delta \right) \left( 1 + \frac{3 \sqrt{2}}{2} \delta \right).
    \end{equation*}
    Now, suppose for a contradiction that $PD(\mathcal{P}) \geq 1.97$. Then we know that  
    \begin{equation*}
            \frac{GridSum}{OPT_{avg}(\phi(\mathcal{P}))} \geq \frac{1.97}{\left( 1 + (4\sqrt{2}) \delta \right) \left( 1 + \left(\frac{3 \sqrt{2}}{2}\right) \delta \right)} \\
    \end{equation*}
So for $\delta:= 1/75$, we get that $$\frac{GridSum}{OPT_{avg}(\phi(\mathcal{P}))} \geq 1.781.$$  But $\phi(\mathcal{P}) \subset G$, and computer analysis, searching over all possible configurations of points on $G$, reveals no such configuration.
%-----------------------------------------------------%
%-----------------------------------------------------%

\subsection*{Proof of Theorem~\ref{theorem:oligarchy}}
We analyze Random Oligarchy using a technical lemma proven in~\cite{2-agree}. For clarity of exposition, we restate the lemma.

\begin{lemma}\cite{2-agree}
\label{lemma:2agree}
The Distortion of any randomized mechanism $f$ is less than or equal to $$1+ 2\max_{p \in [0,1]} m_f(p) \frac{1-p}{p}$$ where $m_f(p)$ is the maximum probability of an alternative being output by $f$ if exactly a $p$ fraction of $N$ consider this alternative their top choice.
\end{lemma}  

To apply Lemma~\ref{lemma:2agree}, we need to bound $m_f(p)$: the maximum probability of an alternative being output by $f$ if exactly a $p$ fraction of $N$ consider this alternative their top choice ($f$, in this context, is just the Random Oligarchy mechanism). We will first consider a coarse bound that allows us to upper bound Distortion by $3$ for arbitrary $|\Se|$, and will then consider a fine grained bound that takes $|\Se|$ into account.

It is straightforward from Definition~\ref{def:RO} that Random Oligarchy outputs an alternative $x \in \Se$ only if (i) it draws at least two agents who report $x$ as their favorite, or if (ii) all three agents drawn report different favorite alternatives, and $x$ is among them, in which case it is output with probability $1/3$. Given that exactly a $p$ fraction of $N$ consider such an $x$ as their top choice, it is straightforward to compute (i). A coarse upper bound on the probability of (ii) is $\binom{3}{1} p (1-p)^2$, the probability that the alternative in question is reported exactly once, and something else is reported for the other two draws. This is potentially an overestimate, since these two other alternatives reported must be distinct, but ignoring this detail gives the following coarse bound on $m_f(p)$.
\begin{equation*}
\begin{split}
	m_f(p) & \leq \binom{3}{3} p^3 + \binom{3}{2} p^2 (1-p) + \frac{1}{3} \binom{3}{1} p (1-p)^2 \\
	& = -p^3 + p^2 + p
\end{split}
\end{equation*}  
Inserting this into Lemma~\ref{lemma:2agree} and simplifying shows that the Distortion of Random Oligarchy is at most 3.
\begin{equation*}
\begin{split}
	Distortion(f) & \leq 1 + 2\max_{p \in [0,1]} \left( \left( -p^3 + p^2 + p \right) \frac{1-p}{p} \right) \\
	& = 3
\end{split}
\end{equation*} 

To obtain a more fine grained bound that considers $|\Se|$, note that the probability of (ii) is maximized when the distribution over favorite points of all agents who do not report $x$ as their favorite is uniform over the remaining $|\Se|-1$ alternatives. Taking into account that for event (ii) the two other reported alternatives must be unique yields the following bound on $m_f(p)$.
\begin{equation*}
\begin{split}
	m_f(p) & \leq \binom{3}{3} p^3 + \binom{3}{2} p^2 (1-p) \\
	& \;\;\; + \frac{1}{3} \binom{3}{1} p (1-p) \left(1-p - \frac{1-p}{|\Se|-1} \right) \\
	& = -p^3 +p^2 + p - \frac{p(p-1)^2}{|\Se|-1} 
\end{split}
\end{equation*}
Substituting this expression into Lemma~\ref{lemma:2agree} allows us to upper bound the Distortion of Random Oligarchy by $$1 + 2\max_{p \in [0,1]} \left( \left( - p^3 +p^2 + p - \frac{p(p-1)^2}{|\Se|-1}  \right) \frac{1-p}{p} \right).$$ Simplifying yields the theorem statement. %$$Distortion(f) \leq 1 + 2\max_{p \in [0,1]} \left( 1 + p^2(p-2) + \frac{(p-1)^3}{|\Se|-1} \right).$$

%--------------------------------------------------------%

\end{document}